\documentclass[11pt]{amsart}
\usepackage[usenames]{color}
\usepackage{fullpage}
\usepackage{amscd}
\usepackage{amssymb, latexsym}
\usepackage{mathdots}
\usepackage{nicematrix}

\theoremstyle{plain}
\newtheorem{thm}{Theorem}[section]
\newtheorem{theorem}[thm]{Theorem}

\newtheorem{proposition}[thm]{Proposition}

\theoremstyle{definition}
\newtheorem{de}[thm]{Definition}

\newtheorem{example}[thm]{Example}

\newtheorem{algorithm}[thm]{Algorithm}

\newcommand{\nie}[1]{{\color{blue}{#1}\color{black}{}}}

\numberwithin{equation}{section}
\newcommand{\N}{\mathbb{N}}

\begin{document}

\title{Retracts of degenerate solutions of \\the Yang-Baxter equation}

\author{P\v remysl Jedli\v cka}
\author{Agata Pilitowska}

\address{(P.J.) Department of Mathematics and Physics, Faculty of Engineering, Czech University of Life Sciences, Kam\'yck\'a 129, 16521 Praha 6, Czech Republic}
\address{(A.P.) Faculty of Mathematics and Information Science, Warsaw University of Technology, Koszy\-kowa 75, 00-662 Warsaw, Poland}

\email{(P.J.) jedlickap@tf.czu.cz}
\email{(A.P.) agata.pilitowska@pw.edu.pl}



\date{\today}

\begin{abstract}
Most of the set-theoretical solutions of the Yang-Baxter equation studied in the past years were non-degenerate multipermutation solutions. For degenerate solutions, a correct definition of multipermutation solutions has not been established so far. We fill here this gap providing a definition
of multipermutation solutions that generalizes the one for non-degenerate solutions and we find an axiomatic description of this class by a set of equations that generalizes the equations describing non-degenerate multipermutation solutions.

It turned out that the results do not need all the properties of solutions of the Yang-Baxter equation and therefore we prove them in a general universal-algebraic setting.
\end{abstract}

\maketitle

\section{Introduction}

The Yang-Baxter equation is a fundamental equation occurring in mathematical physics. It appears, for example, in integrable models in statistical mechanics, quantum field theory or Hopf algebras~(see e.g. \cite{Jimbo, K}). Searching for its solutions has been absorbing researchers for many years.

Let us recall that, for a vector space $V$, a {\em solution of the Yang--Baxter equation} is a linear mapping $r:V\otimes V\to V\otimes V$ 
 such that
\begin{align*}
(id\otimes r) (r\otimes id) (id\otimes r)=(r\otimes id) (id\otimes r) (r\otimes id).
\end{align*}

Description of all possible solutions seems to be extremely difficult and therefore
there were some simplifications introduced by Drinfeld in \cite{Dr90}.
Let  $X$ be a basis of the space $V$ and let $\sigma:X^2\to X$ and $\tau: X^2\to X$ be two mappings. We say that $(X,\sigma,\tau)$ is a {\em set-theoretic solution of the Yang--Baxter equation} if
the mapping 
$$x\otimes y \mapsto \sigma(x,y)\otimes \tau(x,y)$$ extends to a solution of the Yang--Baxter
equation. It means that $r\colon X^2\to X^2$, where $r=(\sigma,\tau)$,  satisfies the \emph{braid relation}:
\begin{equation}\label{eq:braid}
(id\times r)(r\times id)(id\times r)=(r\times id)(id\times r)(r\times id).
\end{equation}

A solution is called {\em left non-degenerate} if the mappings $\sigma_x=\sigma(x,\_)$  are bijections
	and {\em right non-degenerate} if the mappings $\tau_y=\tau(\_\,,y)$
		are bijections,
for all $x,y\in X$. A solution is {\em non-degenerate} if it is left and right non-degenerate.
A~solution is called {\em bijective} if $r$ is a~bijection of~$X^2$. In particular, it is {\em involutive} if $r^2=\mathrm{id}_{X^2}$. All solutions $(X,\sigma,\tau)$, we study in this paper, are set-theoretic, so we will call them simply \emph{solutions}. The set $X$ can be of arbitrary cardinality.

If $(X, \sigma,\tau)$ is a solution then directly, by the braid relation, we obtain, for $x,y,z\in X$:
\begin{align}
\sigma_x\sigma_y&=\sigma_{\sigma_x(y)}\sigma_{\tau_y(x)} \label{birack:1}\\
\tau_{\sigma_{\tau_y(x)}(z)}\sigma_x(y)&=\sigma_{\tau_{\sigma_y(z)}(x)}\tau_{z}(y) \label{birack:2}\\
\tau_x\tau_y&=\tau_{\tau_x(y)}\tau_{\sigma_y(x)} \label{birack:3}
\end{align}                                         

Some of the results presented here actually do not depend on these axioms and hence we shall be working in a broader universal algebraic setting.
Let $(X,\Gamma)$ be an algebra, where $\Gamma$ is a set of binary operations. An equivalence relation $\mathord{\asymp}\subseteq X\times X$ such that for $x_1,x_2,y_1,y_2\in X$ 
\begin{align}\label{congr}
&x_1\asymp x_2\;\; {\rm and} \;\; y_1\asymp y_2\quad \Rightarrow\quad \gamma(x_1,y_1)\asymp \gamma(x_2,y_2)\quad \text{for all }\gamma\in\Gamma
\end{align}
is called a \emph{congruence} of the algebra $(X,\Gamma)$
A congruence induces a quotient solution $(X^{\asymp},\Gamma)$ on its classes with $\gamma(x^{\asymp},y^{\asymp})=\gamma(x,y)^{\asymp}$, for $\gamma\in\Gamma$, $x^{\asymp},y^{\asymp}\in X^{\asymp}$  and $x\in x^{\asymp},\; y\in y^{\asymp}$.

In \cite{ESS} Etingof, Schedler and Soloviev introduced, for each non-degenerate involutive solution $(X,\sigma,\tau)$, the equivalence relation $\sim$ on the set $X$: for each $x,y\in X$
\begin{align}\label{rel:sim}
x\sim y\quad \Leftrightarrow\quad \tau_x=\tau_y
\end{align}
and they showed that 
$\sim$ is a congruence of the solution and the factor solution is again non-degenerate.
In the case of a non-degenerate non-involutive solution $(X,\sigma,\tau)$,
such an equivalence relation $\sim$
does not need to 
be a congruence. It is so if the solution is right distributive (see \cite[Theorem 3.4]{JPZ20b}) 
but, without right distributivity, a suitable definition of a congruence turns out to be
\begin{equation}\label{eq:retraction}
 x\sim y 
 \quad \Leftrightarrow \quad {\sigma_x=\sigma_y} \wedge {\tau_x=\tau_y}.
\end{equation}

Lebed and Vendramin showed in \cite{LV} that  the relation $\sim$ is a congruence of injective solutions. 
In \cite{JPZ19} the authors together with Zamojska-Dzienio generalized the result for any non-degenerate solution $(X,\sigma,\tau)$.  A substantially shorter proof has appeared in \cite{CJKAV} for non-degenerate bijective solutions.

The quotient solution $\mathrm{Ret}(X):=(X^{\sim},\sigma,\tau)$ is called the \emph{retraction} solution of $(X,\sigma,\tau)$. 
If the equivalence $\sim$ is identity, 
we call the solution {\em irretractable.} Otherwise, one can define the \emph{iterated retraction} in the following way: ${\rm Ret}^0(X):=(X,\sigma,\tau)$ and
	${\rm Ret}^k(X):={\rm Ret}({\rm Ret}^{k-1}(X))$, for any natural number $k>1$. 
If there exists an integer~$k$ such that $\mathrm{Ret}^k(X)$ has one element only then we say that
$(X,\sigma,\tau)$ has {\em multipermutation level $k$}. 
Gateva-Ivanova studied such solutions in many papers. In particular, together with Cameron they gave in \cite{GIC12} an equational characterization of square-free non-degerate involutive solutions of multipermutation level at most $k$. Recently, we extended  in \cite{JP25} their result and gave an equational characterization of all non-degenerate multipermutation solutions.

In recent years, degenerate solutions started to attract researchers (see e.g. \cite{CMMS} -- \cite{CV},  
\cite{R19}), 
and this, naturally, raised the question whether
some generalization of the multipermutation level can be defined as well.
Actually, the constructions that appear in the literature are either irretractable, or, by definition, they inherit some iterated structure from the associative algebras they are derived from. Therefore the question of
retracts of these specific solutions can be answered immediately.
However, in the full generality, there exist solutions where~$\sim$
is not a~congruence, as we show in Example \ref{ex:bad}.

In the paper we propose a modified definition which generalizes the definition of~$\sim$ for non-degenerate solutions. Additionally, we prove that there exists an axiomatic description of multipermutation level, similar to the axioms we found for non-degenerate solutions~\cite{JP25}.
It turns out that the results do not need all the axioms of solutions of the Yang-Baxter equation, and therefore we prove them in a general universal-algebraic setting. The results actually follow quite straightforwardly from the definitions, more interesting than proofs are therefore several examples that justify our choice of these specific definitions over others.

The paper is organized as follows. 
In Section 2 we introduce a general description of the notion of a retract of any algebra with a set of binary operations and provide examples of solutions which justify the described definitions. In Section 3 we introduce an equational basis for multipermutation level. We provide an example showing that a simplified version, that works for non-degenerate solutions, does not suffice here, except for solutions of the multipermutation level 2.
Some of the examples were found using Mace4~\cite{Prover}.

\section{Retract equivalence}

In this section we present a general description of the notion of a retract and we will present an algorithm how to find the retract.
Throughout the section
we suppose that we have an algebra $(X,\Gamma)$ where $\Gamma$ is a set of binary operations.
For any $a,b\in X$ and each binary operation $\gamma\in \Gamma$, the image of a pair $(a,b)$ 
will be denoted as $\gamma_a(b)$. Hence we will consider $\gamma$ as a
mapping $X\to X^X$ rather than $X\times X\to X$ (the concept known as ``currying'' in a cartesian closed category).

The notion of a congruence has been already introduced before but here we shall pay more attention on formal details.
Let us recall that for $(X,\Gamma)$  an equivalence relation $\mathord{\asymp}\subseteq X\times X$ is \emph{compatible} with a binary operation $\gamma\in X^X$ if for $x_1,x_2,y_1,y_2\in X$,
\begin{equation}
	\gamma_{x_1}(y_1)\; \asymp \; \gamma_{x_2}(y_2),
\end{equation}
whenever $x_1\asymp x_2$ and $y_1\asymp y_2$.
This condition ensures that there exists a well defined operation
$\gamma/_{\asymp}$ on the equivalence classes of~$\asymp$.
If an equivalence is compatible with all the operations $\gamma\in \Gamma$ then it is called a {\em congruence} of $(X,\Gamma)$.

\begin{de}
Let $(X,\Gamma)$ be an algebra with a set $\Gamma$ of binary operations. We define the equivalence $\sim$ on~$X$ as
\begin{equation}
    x\sim y \qquad\text{ if and only if}\qquad
    \gamma_x=\gamma_y,\text{ for all }\gamma\in\Gamma.
\end{equation}
Let $\approx$ be the biggest congruence of~$(X,\Gamma)$ smaller or equal than $\sim$. 

We define the {\em retract} $\mathrm{Ret}(X)$ as the factor $(X,\Gamma)/\mathord{\approx}$ and the {\em iterated retraction} $\mathrm{Ret}^{i+1}(X)$ as $\mathrm{Ret}^{i}(X)/\mathord{\approx}$, for $i\in\mathbb{N}$. If there exists an integer~$k$ such that $\mathrm{Ret}^k(X)$ has one element only then we will say that
the algebra $(X,\Gamma)$ has {\em multipermutation level $k$}. We will say that an algebra $(X,\Gamma)$ is \emph{irretractable} if ${\rm Ret}(X)=(X,\Gamma)$.
\end{de}

It is well known that the ordered set of all congruences of an algebraic structure forms a sublattice of the lattice of all equivalences and this lattice is complete (see e.g. \cite[Theorem 5.3.]{BS}). Hence the congruence~$\approx$ always exists.

There is a standard algorithm for finding the biggest congruence under an equivalence. It is used, for instance, for finding the minimal finite automaton equivalent to a given automaton (see e.g. \cite{HMU}).

\begin{algorithm} 
Let $(X,\Gamma)$ be an algebra with a set $\Gamma$ of binary operations and let $a,b\in X$. We inductively construct relations $\sim_i$, for any $i\in\N\cup\{\infty\}$:
\begin{itemize}
    \item Initial step: $\mathord{\sim_0}\nie{:}=\mathord{\sim}$.
    \item Inductive step: $a\sim_{i+1} b$ if and only if
    $a\sim_i b$ and $\gamma_x(a)\sim_i\gamma_x(b)$, for all $x\in X$, $\gamma\in\Gamma$ and $i\in\N$.
    \item Limit step: $a\sim_\infty b$ if and only if $a\sim_i b$, for all $i\in\N$.
\end{itemize}
\end{algorithm}

\begin{proposition}
Let $(X,\Gamma)$ be an algebra. Then
\begin{itemize}
    \item Every relation $\sim_i$, for any $i\in \N\cup\{\infty\}$, is an equivalence on $X$.
\item The relation $\sim_\infty$ is a congruence of $(X,\Gamma)$.
\item The relations $\sim_\infty$ and $\approx$ are equal.
\end{itemize}
\end{proposition}

\begin{proof}
    It is evident that, for any $i\in \N\cup\{\infty\}$,  every relation $\sim_i$ is reflexive, symmetric and transitive. 
    We now prove that the relation $\sim_\infty$ is a congruence of $(X,\Gamma)$. 
    Suppose, by contradiction,
    that there exist $a,b,x,y\in X$ and $\gamma\in\Gamma$ such that
    $a\sim_\infty b$, $x\sim_\infty y$ and $\gamma_x(a)\not\sim_\infty\gamma_y(b)$.
    Then there exists $i\in\N$ such that $\gamma_x(a)\not\sim_i\gamma_y(b)$
    and hence $\gamma_x(a)\not\sim_i\gamma_x(b)$. Therefore, it forces 
    $a\not\sim_{i+1} b$, a contradiction.

    Further assume, 
    by contradiction, $\mathord{\sim_\infty}\not\supseteq \mathord{\approx}$. Then there exist $a,b\in X$ such that $a\approx b$ and $a\not\sim_i b$, for some $i\in\N$. Suppose that $i$ is the minimal number that such~$a$ and $b$ exist. Clearly $i>0$. By minimality of~$i$, $a\sim_{i-1} b$ and there exists $\gamma\in\Gamma$ such that $\gamma_x(a)\not\sim_{i-1}\gamma_x(b)$, for some~$x\in X$. Since~$\approx$ is a congruence, we have $\gamma_x(a)\approx\gamma_x(b)$, a contradiction with the minimality of~$i$. Hence $\mathord{\sim_\infty}\supseteq\mathord{\approx}$. The reciproque inclusion is obtained from the maximality of~$\approx$.
\end{proof}

Observe that we can treat a solution $(X,\sigma,\tau)$ as an algebra with two binary operations. Using ``currying'' notation, we consider the operations of this algebra as $\sigma_x:=\sigma(x,\_)$ and $\tau_y:=\tau(\_,y)$, for any $x,y\in X$. Moreover, we can consider any non-degenerate solution as an algebra with four binary operations: $\sigma_x$, $\sigma_x^{-1}$, $\tau_y$ and $\tau_y^{-1}$. In this case, 
we know that already $\sim$ is a congruence of such an algebra and hence $\mathord{\approx}=\mathord{\sim}$. 
Therefore, to present examples of solutions where these two relations differ, we have to
display degenerate solutions. 

\begin{example}\label{ex:bad}
    Consider the following degenerate solution $(X,\sigma,\tau)$ with binary operations defined as follows:
    \[
    \begin{array}{r|ccccc}
    \sigma & a & b & c & d & e\\
    \hline
    a & a & e & d & c & a\\
    b & a & a & d & c & a\\
    c & a & e & d & c & a\\
    d & a & a & d & c & a\\
    e & a & e & d & c & a
    \end{array},\qquad
    \begin{array}{r|ccccc}
    \tau & a & b & c & d & e\\
    \hline
    a & a & a & a & a & a\\
    b & b & b & b & b & b\\
    c & c & c & c & c & c\\
    d & d & d & d & d & d\\
    e & e & e & e & e & e
    \end{array}.
    \]
Then, for $x,y\in X$, $x\sim y$ if and only if $\sigma_x=\sigma_y$, that means the corresponfing rows of the left table are equal, and $\tau_x=\tau_y$, that means the corresponding columns of the right table are equal.
Clearly, there are two classes of the equivalence~$\sim_0$. Let us denote them $I=\{a,c,e\}$ and $II=\{b,d\}$,
and let us write the table of the ``factor solution'' from the point of view of corresponding elements:
\[
\begin{NiceArray}{cc|ccc|cc}
    & \sigma & a & c & e & b & d\\
    \hline  
    \Block{3-1}{I} & a & I & II & I & I & I \\
    & c & I & II & I & I & I\\
    & e & I & II & I & I & I\\
    \hline
    \Block{2-1}{II} & b & I & II & I & I & I \\
    & d & I & II & I & I & I
\end{NiceArray},\qquad
\begin{NiceArray}{cc|ccc|cc}
    & \tau & a & c & e & b & d\\
    \hline  
    \Block{3-1}{I} & a & I & I & I & I & I \\
    & b & I & I & I & I & I\\
    & c & I & I & I & I & I\\
    \hline
    \Block{2-1}{II} & b & II & II & II & II & II \\
    & d & II & II & II & II & II
\end{NiceArray}.
\]
The equivalence $\sim_0$ would be a congruence if, in every block, the colmuns of $\sigma$ and the rows of $\tau$ were homogeneous. Since it is not so, we see that
we have to split the class~$I$ into two classes. Let us denote them $III=\{a,e\}$ and $IV=\{c\}$. Now  the relation $\sim_1$ has three equivalence classes:
\[
\begin{NiceArray}{cc|cc|c|cc}
    & \sigma & a & e & c & b & d\\
    \hline  
    \Block{2-1}{III} & a & III & III & II & III & IV \\
    & e & III & III & II & III & IV\\
    \hline
    IV & c & III & III & II & III & IV\\
    \hline
    \Block{2-1}{II} & b & III & III & II & III & IV \\
    & d & III & III & II & III & IV
\end{NiceArray},\qquad
\begin{NiceArray}{cc|cc|c|cc}
    & \tau & a & e & c & b & d\\
    \hline  
    \Block{2-1}{III} & a & III & III & III & III & III \\
    & e & III & III & III & III & III\\
    \hline
    IV & c & IV & IV & IV & IV & IV\\
    \hline
    \Block{2-1}{II} & b & II & II & II & II & II \\
    & d & II & II & II & II & II
\end{NiceArray}.
\]
Again, the relation~$\sim_1$ is still not a congruence ($b\sim_1 d$ but $\sigma_a(b)\not\sim_1 \sigma_a(d)$)
and we have to split the class $II$ into $V=\{b\}$ and $VI=\{d\}$.
\[
\begin{NiceArray}{cc|cc|c|c|c}
    & \sigma & a & e & c & b & d\\
    \hline  
    \Block{2-1}{III} & a & III & III & VI & III & IV \\
    & e & III & III & VI & III & IV\\
    \hline
    IV & c & III & III & VI & III & IV\\
    \hline
    V & b & III & III & VI & III & IV \\
    \hline
    VI & d & III & III & VI & III & IV
\end{NiceArray},\qquad
\begin{NiceArray}{cc|cc|c|c|c}
    & \tau & a & e & c & b & d\\
    \hline  
    \Block{2-1}{III} & a & III & III & III & III & III \\
    & e & III & III & III & III & III\\
    \hline
    IV & c & IV & IV & IV & IV & IV\\
    \hline
    V & b & V & V & V & V & V \\
    \hline
    VI & d & VI & VI & VI & VI & VI
\end{NiceArray}.
\]
We see that the process stops at this stage, that means $\mathord{\sim_2}=\mathord{\sim_3}=\cdots=\mathord{\sim_\infty}=\mathord{\approx}$. The retract $\mathrm{Ret}(X)$ has four elements and the second retract $\mathrm{Ret}^2(X)$ has only one element. Hence, the solution $(X,\sigma,\tau)$ has multipermutation level $2$.
\end{example}

\begin{example}
    Let $Y=\mathbb{N}\cup (\infty-\mathbb{N})=\{0,1,2,3,\ldots\}\cup\{\ldots ,\infty-2,\infty-1,\infty\}$. This ordered set is a chain with the smallest element $0$ and the greatest element $\infty$. Let $X=Y\cup\{\overline{\infty-1}\}$ where $\overline{\infty-1}$ is a twin copy of the element $\infty-1$. Let us now define two binary mappings $\sigma$ and $\tau$ on~$X$ as follows:
    \begin{align*}
    \sigma_{x}(y)&=0 \qquad\text{for all }x,y\in X,\\
    \tau_y(x)&=\begin{cases}
        \overline{\infty-1} &\text{if }x=\infty\text{ and }y=0,\\
        \infty-2 & \text{if }x=\overline{\infty-1},\\
        \max(0,x-1) & \text{otherwise}.\\        
    \end{cases}
    \end{align*}
    It is easy to check that, for all $x,y,z\in X$,
    \[\sigma_x\sigma_y(z)=\sigma_{\sigma_x(y)}\sigma_{\tau_y(x)}(z)=\tau_{\sigma_{\tau_y(x)}(z)}\sigma_x(y)=\sigma_{\tau_{\sigma_y(z)}(z)}\tau_z(y)=0\]
    and
    \[\tau_x\tau_y(z)=\tau_{\tau_x(y)}\tau_{\sigma_y(x)}(z)=\max(0,z-2).\]
    Hence we have a degenerate solution. Now, an easy inductive argument shows that, for any $i\in\mathbb{N}$,
    \[ x\sim_i y \quad\text{ if and only if }\quad
    x=y \text{ or } (x>i\text{ and }y>i).\]
    This shows that 
    the algorithm does not stabilize in any finite step. In the limit step we get
    \[ x\approx y \quad\text{ if and only if }\quad
    x=y \text{ or } (x\notin\mathbb{N}\text{ and }y\notin\mathbb{N}).\]
\end{example}

\section{Axioms}
In \cite[Theorem 5.15]{GIC12} Gateva-Ivanova characterized non-degenerate involutive multipermutation solutions by 
one equation and to describe the equation, she introduced an expression called \emph{a tower of actions} \cite[Definition 5.9]{GIC12}. However, to obtain in \cite{JP25} an equational characterization of non-degenerate non-involutive  multipermutation solutions we needed much more equations of this type. Which is why we decided to define some iterative terms that we shall generalize here.

Let $(X,\Gamma)$ be an algebra with a set $\Gamma$ of binary operations.
Let $\Sigma=\{\gamma_x\mid \gamma\in\Gamma, x\in X\}$
and let $\Sigma^*$ be the monoid of all words over~$\Sigma$.
Let us define, for each $i\in \mathbb{N}$, for all elements $x,z_1,\ldots,z_i\in X$,
for all mappings $\gamma^{(1)},\gamma^{(2)},\ldots,\gamma^{(i-1)},\gamma^{(i)}\in \Gamma$,
and for all words $\mathfrak{w}_1,\ldots,\mathfrak{w}_k\in \Sigma^*$
\begin{align*}
&\Omega_0(x):=x,\\
&\Omega_1(\mathfrak{w}_1\gamma^{(1)},x,z_1):=\mathfrak{w}_1\gamma^{(1)}_x(z_1), \\
&\vdots\\
&\Omega_i(\mathfrak{w}_i\gamma^{(i)},\ldots,\mathfrak{w}_1\gamma^{(1)},x,z_1,\ldots,z_{i-1},z_i):=\mathfrak{w}_i\gamma^{(i)}_{\Omega_{i-1}(\mathfrak{w}_{i-1}\gamma^{(i-1)},\ldots,\mathfrak{w}_1\gamma^{(1)},x,z_1,\ldots,z_{i-1})}(z_i).
\end{align*}

\begin{de}
    An algebra $(X,\Gamma)$ is called $k$-permutational, if, for all mappings $\gamma^{(1)},\ldots,\gamma^{(k)}\in \Gamma$, all elements $x_1,\ldots,x_k, y\nie{, z}\in X$, and all words $\mathfrak{w}_1,\ldots,\mathfrak{w}_k\in\Sigma^*$
    \begin{equation}\label{eq:permut}
        \Omega_k(\mathfrak{w}_k\gamma^{(k)},\ldots,\mathfrak{w}_1\gamma^{(1)},y,x_1,\ldots,x_k)=
        \Omega_k(\mathfrak{w}_k\gamma^{(k)},\ldots,
        \mathfrak{w}_1\gamma^{(1)},z,x_1,\ldots,x_k).
    \end{equation}
\end{de}

Note that in \cite{JP25} a different definition was given,
not using the words $\mathfrak{w}_i$. 
This definition seems to work well when decribing algebras where
every mapping from~$\Sigma$ is a bijection. 
However,
in the general case, without the words $\mathfrak{w}_i$ we may obtain irretractable algebras, as we shall see on the following example.

\begin{example}\label{exm:3-perm}
    Let $n>2$, let $X=\{0,1,\ldots,n\}$ and let
    \[
    \sigma_{x}(y)=\begin{cases} \min\{n,x+1\} & \text{if } y=0,\\
    1 & \text{if }y\neq 0,
    \end{cases}
    \qquad
    \tau_y(x)=\begin{cases}
        2 &\text{if }x=n\text{ and }y=0,\\
        1 &\text{if }x\neq n\text{ or }y\neq 0.
    \end{cases}
    \]
Clearly, for all $x,y,z\in X$,
\[\sigma_x\sigma_y(z)=\sigma_{\sigma_x(y)}\sigma_{\tau_y(x)}(z)=\tau_{\sigma_{\tau_y(x)}(z)}\sigma_x(y)=\sigma_{\tau_{\sigma_y(z)}(z)}\tau_z(y)=
\tau_x\tau_y(z)=\tau_{\tau_x(y)}\tau_{\sigma_y(x)}(z)=1.\]
Hence $(X,\sigma,\tau)$ 
is a solution which is obviously degenerate. Moreover, $n\sim (n-1)$ is the only pair of different elements in the relation~$\sim$. But $\tau_0(n)=2\neq 1=\tau_0(n-1)$ and hence the 
solution is irretractable. Furthermore
\[
\Omega_{n-1}(\sigma,\ldots,\sigma,0,0,\ldots,0)=n-1\neq n=\Omega_{n-1}(\sigma,\ldots,\sigma,1,0,\ldots,0)
\]
and the solution is not $n-1$-permutational.
The reader can, by a case-by-case analysis, prove that
\[
        \Omega_n(\gamma^{(n)},\ldots,\gamma^{(1)},y,x_1,\ldots,x_n)=
        \Omega_n(\gamma^{(n)},\ldots,
        \gamma^{(1)},z,x_1,\ldots,x_n),
    \]
for all $\gamma^{(1)},\ldots,\gamma^{(n)}\in\Gamma=\{\sigma,\tau\}$
and $x_1,\ldots,x_n,y,z\in X$.
However, the solution is not $n$-permutational,
according to the definition \eqref{eq:permut} given in this paper, because
\[\Omega_n(\sigma,\tau_0\sigma,\sigma,\ldots,\sigma,0,0,\ldots,0)=2\neq 3=\Omega_n(\sigma,\tau_0\sigma,\sigma,\ldots,\sigma,1,0,\ldots,0).\]
\end{example}

In the case of non-degenerate solutions, the authors showed in~\cite{JP25} that there is a strong connection between $k$-permutational solutions (understood according to the definition for non-degenerate solutions) and the $k$ retract with respect to the relation $\sim$.  A similar relationship holds for algebras with binary operations.
\begin{theorem}\label{thm:main}
    An algebra $(X,\Gamma)$, with a set $\Gamma$ of binary operations, is $k$-permutational if and only if it has multipermutation level at most~$k$.
\end{theorem}

\begin{proof}
    We prove the claim by an induction on~$k\geq 0$. The claim is true for $k=0$.

    ``$\Rightarrow$'': Suppose that every $k$-permutational solution has multipermutation level at most~$k$ and suppose that $(X,\Gamma)$ is $(k+1)$-permutational. We want to prove that $\mathrm{Ret}(X)$ is $k$-permutational.

    Consider $\gamma^{(1)},\ldots,\gamma^{(k)}\in\Gamma$,
    $\mathfrak{w}_1,\ldots,\mathfrak{w}_k\in\Sigma^*$ and $x_1,\ldots x_k, y, z \in X$ arbitrary. We have, for all $\delta\in \Gamma$ and $w\in X$, the following:    
    \[
    \Omega_1(\delta,\Omega_k(\mathfrak{w}_k\gamma^{(k)},\ldots,\mathfrak{w}_1\gamma^{(1)},y,x_1,\ldots,x_k),w)=
    \Omega_1(\delta,\Omega_k(\mathfrak{w}_k\gamma^{(k)},\ldots,\mathfrak{w}_1\gamma^{(1)},z,x_1,\ldots,x_k),w).
    \]
    This implies 
    \[
    \Omega_k(\mathfrak{w}_k\gamma^{(k)},\ldots,\mathfrak{w}_1\gamma^{(1)},y,x_1,\ldots,x_k) \sim \Omega_k(\mathfrak{w}_k\gamma^{(k)},\ldots,\mathfrak{w}_1\gamma^{(1)},z,x_1,\ldots,x_k).
    \]
    
    Suppose now, by contradiction,
    that there exist $\alpha^{(1)},\ldots,\alpha^{(k)}\in\Gamma$, 
    $\mathfrak{a}_1,\ldots,\mathfrak{a}_k\in\Sigma^*$ and
    $a_1,\ldots,a_k,\\b, c\in X$ such that, for some $i\in \N$ 
    \[
    \Omega_{k}(\mathfrak{a}_k\alpha^{(k)},\ldots,\mathfrak{a}_1\alpha^{(1)},b,a_1,\ldots,a_k) \not\sim_i \Omega_{k}(\mathfrak{a}_k\alpha^{(k)},\ldots,\mathfrak{a}_1\alpha^{(1)},c,a_1,\ldots,a_k),
    \]
    and suppose that $i$ is minimal. Then there exist $\beta\in\Gamma$ and~$d\in X$ such that
    \begin{align}\label{eq:not}
    \beta_d(
    \Omega_{k}(\mathfrak{a}_k\alpha^{(k)},\ldots,\mathfrak{a}_1\alpha^{(1)},b,a_1,\ldots,a_k)) \not\sim_{i-1} \beta_d(\Omega_k(\mathfrak{a}_k\alpha^{(k)},\ldots,\mathfrak{a}_1\alpha^{(1)},c,a_1,\ldots,a_k)).
    \end{align}
Since for any $x\in X$
\[
    \beta_d(
    \Omega_{k}(\mathfrak{a}_k\alpha^{(k)},\ldots,\mathfrak{a}_1\alpha^{(1)},x,a_1,\ldots,a_k)) =\beta_d\mathfrak{a}_k\alpha^{(k)}_{
    \Omega_{k-1}(\alpha^{(k-1)},\ldots,\mathfrak{a}_1\alpha^{(1)},x,a_1,\ldots,x_{k-1})}(a_k),
    \]  
the equality \eqref{eq:not} means
    \[
    \Omega_{k}(\beta_d\mathfrak{a}_k
    \alpha^{(k)},\ldots,\mathfrak{a}_1\alpha^{(1)},b,a_1,\ldots,a_k) \not\sim_{i-1} \Omega_k(\beta_d\mathfrak{a}_k\alpha^{(k)},\ldots,\mathfrak{a}_1\alpha^{(1)},c,a_1\ldots,a_k,d)
    \]
    a contradiction with the minimality of~$i$.
    Hence \[
    \Omega_k(\mathfrak{w}_k\gamma^{(k)},\ldots,\mathfrak{w}_1\gamma^{(1)},y,x_1,\ldots,x_k) \approx \Omega_k(\mathfrak{w}_k\gamma^{(k)},\ldots,\mathfrak{w}_1\gamma^{(1)},z,x_1,\ldots,x_k),
    \] and the retract is $k$-permutational.
    By the induction hypothesis, the multipermutation level of the retract is at most~$k$. Hence the multipermutation level of $(X,\Gamma)$ is at most $k+1$.

    ``$\Leftarrow$'': Now suppose that every algebra with a set of binary operations and 
of multipermutation level at most $k$ is $k$-permutational. 
    Let us have an algebra $(X,\Gamma)$ of multipermutation level at most~$k+1$.
    Then, for all $\gamma^{(1)},\ldots,\gamma^{(k)}\in\Gamma$, $\mathfrak{w}_1,\ldots,\mathfrak{w}_k\in\Sigma^*$ and $x_1,\ldots,x_k, y, z\in X$ 
    \[
    \Omega_k(\mathfrak{w}_k\gamma^{(k)},\ldots,\mathfrak{w}_1\gamma^{(1)},y,x_1,\ldots,x_k) \approx \Omega_k(\mathfrak{w}_k\gamma^{(k)},\ldots,\mathfrak{w}_1\gamma^{(1)},z,x_1,\ldots,x_k).
    \]
In particular this means that,    
for any $\gamma\in\Gamma$
\[
    \gamma_{\Omega_k(\mathfrak{w}_k\gamma^{(k)},\ldots,\mathfrak{w}_1\gamma^{(1)},y,x_1,\ldots,x_k)}=\gamma_{\Omega_k(\mathfrak{w}_k\gamma^{(k)},\ldots,\mathfrak{w}_1\gamma^{(1)},z,x_1,\ldots,x_k)}.
    \]
Hence, for $\gamma^{k+1}\in\Gamma$, $\mathfrak{w}_{k+1}\in\Sigma^*$ and $x_{k+1}\in X$, we obtain
    \begin{multline*}
        \Omega_{k+1}(\mathfrak{w}_{k+1}\gamma^{(k+1)},\ldots,\mathfrak{w}_1\gamma^{(1)},y,x_1,\ldots,x_{k+1})=\\
        \mathfrak{w}_{k+1}\gamma^{(k+1)}_{\Omega_{k}(\mathfrak{w}_k\gamma^{(k)},\ldots,\mathfrak{w}_1\gamma^{(1)},y,x_1,\ldots,x_k)}(x_{k+1})=\\
       \mathfrak{w}_{k+1}\gamma^{(k+1)}_{\Omega_{k}(\mathfrak{w}_k\gamma^{(k)},\ldots,\mathfrak{w}_1\gamma^{(1)},z,x_1,\ldots,x_k)}(x_{k+1})=\\
        \Omega_{k+1}(\mathfrak{w}_{k+1}\gamma^{(k+1)},\ldots,\mathfrak{w}_1\gamma^{(1)},z,x_1,\ldots,x_{k+1})
    \end{multline*}
    and $(X,\Gamma)$ is $k+1$ permutational.
\end{proof}

Example~\ref{exm:3-perm} works for $n$ at least~$3$ and the reader may wonder why not less. Of course, the word $\mathfrak{w}_1$ is not needed when describing solutions $(X,\sigma,\tau)$ of multipermutation level~$1$ as such solutions satisfy
\[\sigma_{x}=\sigma_y \quad \text{ and }\quad \tau_x=\tau_y,\]
for all $x,y\in X$. The next result says that, in the case of degenerate solutions of multipermutation level~$2$, we can also work without the words $\mathfrak{w}_1$ and $\mathfrak{w}_2$.

\begin{proposition}
        A solution $(X,\sigma,\tau)$ 
    is of multipermutation level at most~$2$ if and only if, for all $x,y,z\in X$:
    \begin{align*}
        \sigma_{\sigma_x(z)}&=\sigma_{\sigma_y(z)}, &
        \tau_{\tau_x(z)}&=\tau_{\tau_y(z)},&
        \sigma_{\tau_x(z)}&=\sigma_{\tau_y(z)}, &
        \tau_{\sigma_x(z)}&=\tau_{\sigma_y(z)}.
    \end{align*}
\end{proposition}

\begin{proof}
    One direction follows trivially from Theorem~\ref{thm:main}, let us hence prove the other one. The assmuption can be translated as $\gamma_x(z)\sim\gamma_y(z)$, for all $\gamma\in\Gamma$ and $x,y,z\in X$. We shall now prove, by an induction on $k\in\N$, that, for all $\gamma^{(1)},\ldots,\gamma^{(k)}\in \Gamma$, $x_1,\ldots,x_k$, $y_1,\ldots,y_k$ and $z\in X$,
    \begin{equation}
        \gamma^{(1)}_{x_1}\gamma^{(2)}_{x_2}\cdots\gamma^{(k)}_{x_k}(z)\sim
        \gamma^{(1)}_{y_1}\gamma^{(2)}_{y_2}\cdots\gamma^{(k)}_{y_k}(z).
    \end{equation}
    The claim is assumed to be true for $k=1$, suppose hence $k>1$.

    \paragraph{Case 1:} Suppose $\gamma^{(k)}=\sigma$ and $\gamma^{(k+1)}=\sigma$. Then
    \begin{multline*}
        \gamma^{(1)}_{x_1}\cdots\gamma^{(k+1)}_{x_{k+1}}(z)=
        \gamma^{(1)}_{x_1}\cdots\gamma^{(k-1)}_{x_{k-1}}\sigma_{x_k}\sigma_{x_{k+1}}(z)\stackrel{\eqref{birack:1}}=
        \gamma^{(1)}_{x_1}\cdots\gamma^{(k-1)}_{x_{k-1}}\sigma_{\sigma_{x_k}(x_{k+1})}\sigma_{\tau_{x_{k+1}}(x_k)}(z)=\\  
        \gamma^{(1)}_{x_1}\cdots\gamma^{(k-1)}_{x_{k-1}}\sigma_{\sigma_{x_k}(x_{k+1})}(\sigma_{\tau_{y_{k+1}}(x_k)}(z))\sim
        \gamma^{(1)}_{x_1}\cdots\gamma^{(k-1)}_{x_{k-1}}\sigma_{\sigma_{x_k}(y_{k+1})}(\sigma_{\tau_{y_{k+1}}(x_k)}(z))=\\
        \gamma^{(1)}_{x_1}\cdots\gamma^{(k-1)}_{x_{k-1}}\sigma_{x_k}\sigma_{y_{k+1}}(z)\sim
        \gamma^{(1)}_{y_1}\cdots\gamma^{(k-1)}_{y_{k-1}}\sigma_{y_k}\sigma_{y_{k+1}}(z)
        =\gamma^{(1)}_{y_1}\cdots\gamma^{(k+1)}_{y_{k+1}}(z).
    \end{multline*}

    \paragraph{Case 2:} Suppose $\gamma^{(k)}=\tau$ and $\gamma^{(k+1)}=\sigma$. Then
    \begin{multline*}
        \gamma^{(1)}_{x_1}\cdots\gamma^{(k+1)}_{x_{k+1}}(z)=
        \gamma^{(1)}_{x_1}\cdots\gamma^{(k-1)}_{x_{k-1}}\tau_{x_k}\sigma_{x_{k+1}}(z)\sim
        \gamma^{(1)}_{x_1}\cdots\gamma^{(k-1)}_{x_{k-1}}\tau_{\sigma_{\tau_z(x_{k+1})}(c)}\sigma_{x_{k+1}}(z)\stackrel{\eqref{birack:2}}=\\
        \gamma^{(1)}_{x_1}\cdots\gamma^{(k-1)}_{x_{k-1}}
        \sigma_{\tau_{\sigma_z(c)}(x_{k+1})}\tau_{c}(z)\sim
        \gamma^{(1)}_{x_1}\cdots\gamma^{(k-1)}_{x_{k-1}}
        \sigma_{\tau_{\sigma_z(c)}(y_{k+1})}\tau_{c}(z)\stackrel{\eqref{birack:2}}=\\
        \gamma^{(1)}_{x_1}\cdots\gamma^{(k-1)}_{x_{k-1}}\tau_{\sigma_{\tau_z(y_{k+1}})}(c)\sigma_{y_{k+1}}(z)\sim
        \gamma^{(1)}_{y_1}\cdots\gamma^{(k-1)}_{y_{k-1}}\tau_{y_k}\sigma_{y_{k+1}}(z)=
        \gamma^{(1)}_{y_1}\cdots\gamma^{(k+1)}_{y_{k+1}}(z),
    \end{multline*}
    for any $c\in X$. Case 3 and Case 4 are analogous.
    
    Let us now have $\gamma,\delta\in\Gamma$, $\mathfrak{v},\mathfrak{w}\in\Sigma^*$ and $x,y,z,t\in X$. Then
    \[
    \mathfrak{v}\gamma_{\mathfrak{w}\delta_x(z)}(t)=
    \mathfrak{v}\gamma_{\mathfrak{w}\delta_y(z)}(t).
    \]
    simply because $\mathfrak{w}\delta_x(z)\sim\mathfrak{w}\delta_y(z)$.
    Hence, according to Theorem~\ref{thm:main}, the solution is of multipermutation level~$2$.
    \end{proof}

\end{document}